\documentclass{article}
\usepackage[margin=1.5in,nohead]{geometry}

\usepackage{amsmath}	
\usepackage{amsfonts}
\usepackage{amssymb}
\usepackage{hyperref}
\usepackage{array}
\usepackage{siunitx}
\usepackage{wrapfig}
\usepackage{lipsum}
\newcolumntype{C}{>$c<$}
\usepackage{graphicx}
\usepackage{mathtools}
\usepackage[ruled, vlined, linesnumbered]{algorithm2e}
\usepackage{color,soul}
\providecommand{\U}[1]{\protect\rule{.1in}{.1in}}                                                                                           
  
\newtheorem{assumption}{Assumption}

\newtheorem{definition}{Definition}
\newtheorem{example}{Example}
\newtheorem{proof}{Proof}

\newtheorem{lemma}{Lemma}                        

\newtheorem{problem}{Problem}
\newtheorem{proposition}{Proposition}
\newtheorem{remark}{Remark}

\DeclareMathOperator*{\argmax}{arg\,max}

\let\oldIEEEkeywords\IEEEkeywords
\def\IEEEkeywords{\oldIEEEkeywords\normalfont\bfseries\ignorespaces}
\begin{document}
	
\title{Graph Temporal Logic Inference for Classification and Identification}

\author{Zhe~Xu, Alexander J Nettekoven, A. Agung~Julius, Ufuk Topcu
	\thanks{Zhe~Xu is with the Institute
		for Computational Engineering and Sciences (ICES), University of Texas,
		Austin, Austin, TX 78712, Alexander Nettekoven is with the Walker Department of Mechanical Engineering at the University of Texas at Austin, 
		Austin, TX 78712, A. Agung~Julius is with the Department of Electrical, Computer, and Systems Engineering, Rensselaer Polytechnic Institute, Troy, NY 12180, Ufuk Topcu is with the Department
		of Aerospace Engineering and Engineering Mechanics, and the Institute
		for Computational Engineering and Sciences (ICES), University of Texas,
		Austin, Austin, TX 78712, e-mail: zhexu@utexas.edu, nettekoven@utexas.edu, juliua2@rpi.edu, utopcu@utexas.edu.}
} 
\maketitle

\begin{abstract}               
	Inferring \textit{spatial-temporal properties} from data is important for many complex systems, such as additive manufacturing systems, swarm robotic systems and biological networks. Such systems can often be modeled as a labeled graph where labels on the nodes and edges represent relevant measurements such as temperatures and distances. We introduce \textit{graph temporal logic} (GTL) which can express properties such as ``whenever a node's label is above 10, for the next 3 time units there are always at least two neighboring nodes with an edge label of at most 2 where the node labels are above 5''. This paper is a first attempt to infer spatial (graph) temporal logic formulas from data for classification and identification. For classification, we infer a GTL formula that classifies two sets of \textit{graph temporal trajectories} with minimal misclassification rate. For identification, we infer a GTL formula that is \textit{informative} and is satisfied by the \textit{graph temporal trajectories} in the dataset with high probability. The informativeness of a GTL formula is measured by the information gain with respect to given prior knowledge represented by a prior probability distribution. We implement the proposed approach to classify the graph patterns of tensile specimens built from selective laser sintering (SLS) process with varying strengths, and to identify informative spatial-temporal patterns from experimental data of the SLS cooldown process and simulation data of a swarm of robots.    
\end{abstract}                                                                      

\section{Introduction}

Inferring \textit{spatial-temporal properties} from data is important in many applications (e.g., additive manufacturing processes, swarm robotics and biological networks). Consider a powder bed of selective laser sintering (i.e., SLS, one type of additive manufacturing) processes \cite{sam_thesis} modeled as a labeled graph (as shown in Fig. \ref{color_intro}). Each subregion of the powder bed is considered a node of the graph and edges exist
between nodes (subregions) within certain distance. Given the time-varying temperature labels at each node and the distance labels on each edge of the graph, we intend to infer knowledge that can characterize the spatial-temporal patterns that emerge in this process.

The representation of the inferred knowledge should be both interpretable to humans and amenable to rigorous mathematical analysis. Variants of temporal logic and spatial logic can express temporal and spatial patterns in a form that resembles natural language \cite{Kong2017}. Furthermore, such expressions are suitable for verification and controller synthesis. Over the past decade, there has been a growing interest in inferring temporal logic formulas from system trajectories \cite{Kong2017,Asarin2012,Jin13,Bombara2016,zhe2016,Hoxha2017,Neider2018,VazquezChanlatte2018LearningTS}. However, to the best of our knowledge, there has been no work on inferring spatial or spatial temporal logic formulas from data. 

Two different categories exist for inferring such spatial or temporal logics from data:  \textit{classification} and \textit{identification}. Given two sets of data, the classification problem is about constructing spatial temporal logic formulas that can classify these two sets of data with
minimal misclassification rate. The identification problem is about identifying spatial temporal logic formulas
that best \textit{fit} one set of data. 

For identification, one measure of the quality of the inferred formula is its \textit{informativeness}, i.e., the extent to which the inferred formula deviates from prior knowledge. In the example as shown in Fig. \ref{color_intro}, suppose that we are given two candidate formulas: one reads as ``for every node, either it is red or it is not red'' and the other one reads as ``for every blue node, there exist at least two red nodes among the neighbors of its neighbors with edge labels of at least 2''. While both formulas are consistent with the labeled graph, the first formula is actually a \textit{tautology} and holds for any labeled graph. In comparison, the second formula describes a specific pattern existing in this labeled graph, hence it is considered to be more \textit{informative} than the first formula.                                                                                            

\begin{figure}[th]
	\centering
	\includegraphics[width=10cm]{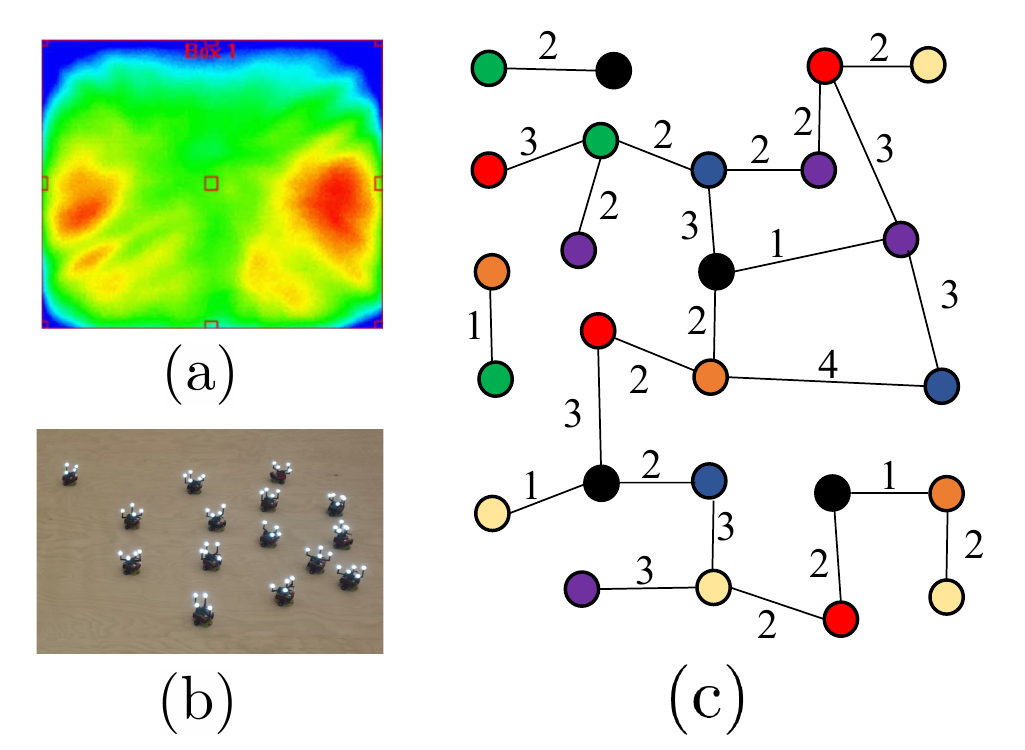}\caption{The powder bed of SLS process recorded with an infrared camera (a) \cite{sam_thesis} and a swarm of mobile robots (b) \cite{Pickem2017} can be both modeled as a labeled graph (c), where the colors indicate node labels (e.g., temperature, probabilistic density) and the numbers indicate edge labels (e.g., distance).}
	\label{color_intro}
\end{figure}                    

In this paper, we first introduce \textit{parametric graph temporal logic} (pGTL), which is an extension of parametric linear temporal logic and focuses on the spatial-temporal properties of the labels on a graph. A pGTL formula has free parameters in the predicates and operators. A \textit{graph temporal logic} (GTL) formula is then induced by a pGTL formula by assigning real values to the parameters of the pGTL formula. We study the following two problems of inferring GTL formulas from spatial-temporal data over a graph:
\begin{itemize}	
   \item \textit{Inferring GTL formulas for classification}:\\        
     We infer a GTL formula that best classifies two sets of \textit{graph-temporal trajectories} (formalized in Sec. \ref{GTL_sec}). 
   \item \textit{Inferring informative GTL formulas for identification}:\\
   We infer a GTL formula that is consistent with a set of \textit{graph-temporal trajectories} and provides a high information gain (formalized in Sec. \ref{info_sec}) over a given prior probability distribution. 
\end{itemize}                                                                                                          
In Sec. \ref{case1}, we implement the classification method to infer GTL formulas that can classify the graph patterns of tensile specimens built from selective laser sintering (SLS) process with varying strengths. In Sec. \ref{case2} and \ref{case3}, we implement the identification method to infer informative GTL formulas from experimental data of the SLS cooldown process and simulation data of a swarm of robots, respectively. 

\noindent\textbf{Related Work.} There exist several spatial (graph) temporal logics in the literature, such as spatial-temporal logic (SpaTeL) in \cite{Haghighi2015} and signal spatio-temporal logic (SSTL) in \cite{Bortolussi2014}. Our proposed GTL is different from both SpaTeL and SSTL as GTL focuses on the propositions on the node labels and edge labels of a graph, and whether there exist certain number of neighbors that satisfy the node propositions with the connecting edges satisfying the edge propositions.  GTL is also different from the logics of graphs in \cite{Courcelle97, Cardelli2002} as they consider logical statements about the structure of the graphs and the changes in the structure, while we consider logical statements about (possibly time-varying) labels that are defined on graphs with fixed structure.                                                        

Our approach of inferring GTL formulas from data is closely related to inferring temporal logic formulas from data. The work in \cite{Kong2017,Bombara2016,Neider2018} focus on inferring temporal logic formulas for classifying two sets of trajectories, while the work in \cite{Asarin2012,zhe2016,VazquezChanlatte2018LearningTS,zhe_info} focus on identifying temporal logic formulas from system trajectories.

\section{Parametric Graph Temporal Logic and Graph Temporal Logic}
\label{GTL_sec}

In this section, we introduce parametric graph temporal logic (pGTL) and graph temporal logic (GTL).           

\subsection{Node and Edge Propositions}
Let $G=(V, E)$ be an undirected graph, where $V$ is a finite set of nodes and $E$ is a finite set of edges. We use $\mathcal{X}$ to denote a (possibly infinite) set of node labels and $\mathcal{Y}$ to denote a (possibly infinite) set of edge labels. We use $s(e)=\{v_1, v_2\}$ to denote the fact that the edge $e\in E$ connects $v_1\in V$ and $v_2\in V$. $\mathbb{T}=\{1, 2, \dots\}$ is a discrete set of time indices. A graph with node labels and edge labels is also called a \textit{labeled graph}.

\begin{definition}
	\label{graph}
	A \textit{graph-temporal trajectory} on a graph $G$ is a tuple $g=(x, y)$, where $x:V\times\mathbb{T}\rightarrow \mathcal{X}$ assigns a node label for each node $v\in V$ at each time index $k\in\mathbb{T}$, and $y:E\times\mathbb{T}\rightarrow\mathcal{Y}$ assigns an edge label for each node $e\in E$ at each time index $k\in\mathbb{T}$. 
\end{definition}                                                            
                
We use $x(v, k)$ to denote the label of node $v$ at time index $k$ and $y(e, k)$ to denote the label of edge $e$ at time index $k$.

\begin{definition}
	\label{node}
	An \textit{atomic node proposition} is a predicate on $\mathcal{X}$, i.e., a Boolean valued map from $\mathcal{X}$. An \textit{edge proposition} is a predicate on $\mathcal{Y}$.
\end{definition}

We use $\pi$ to denote an atomic node proposition, and $\mathcal{O}(\pi)$ to denote the subset of $\mathcal{X}$ for which $\pi$ is true. We use $\rho$ to denote an edge proposition, and $\mathcal{O}(\rho)$ to denote the subset of $\mathcal{Y}$ for which $\rho$ is true. 

We define that a graph-temporal trajectory $g=(x, y)$ satisfies an atomic node proposition $\pi$ at a node $v$ and at a time index $k$, denoted as $(g,v,k)\models\pi$, if and only if $x(v,k) \in \mathcal{O}(\pi)$.
Similarly, we define that a graph-temporal trajectory $g=(x, y)$ satisfies an edge proposition $\rho$ at an edge $e$ and at a time index $k$, denoted as $(g,e,k)\models\rho$, if and only if $y(e,k) \in \mathcal{O}(\rho)$.

\begin{figure}[th]
	\centering
	\includegraphics[width=6cm]{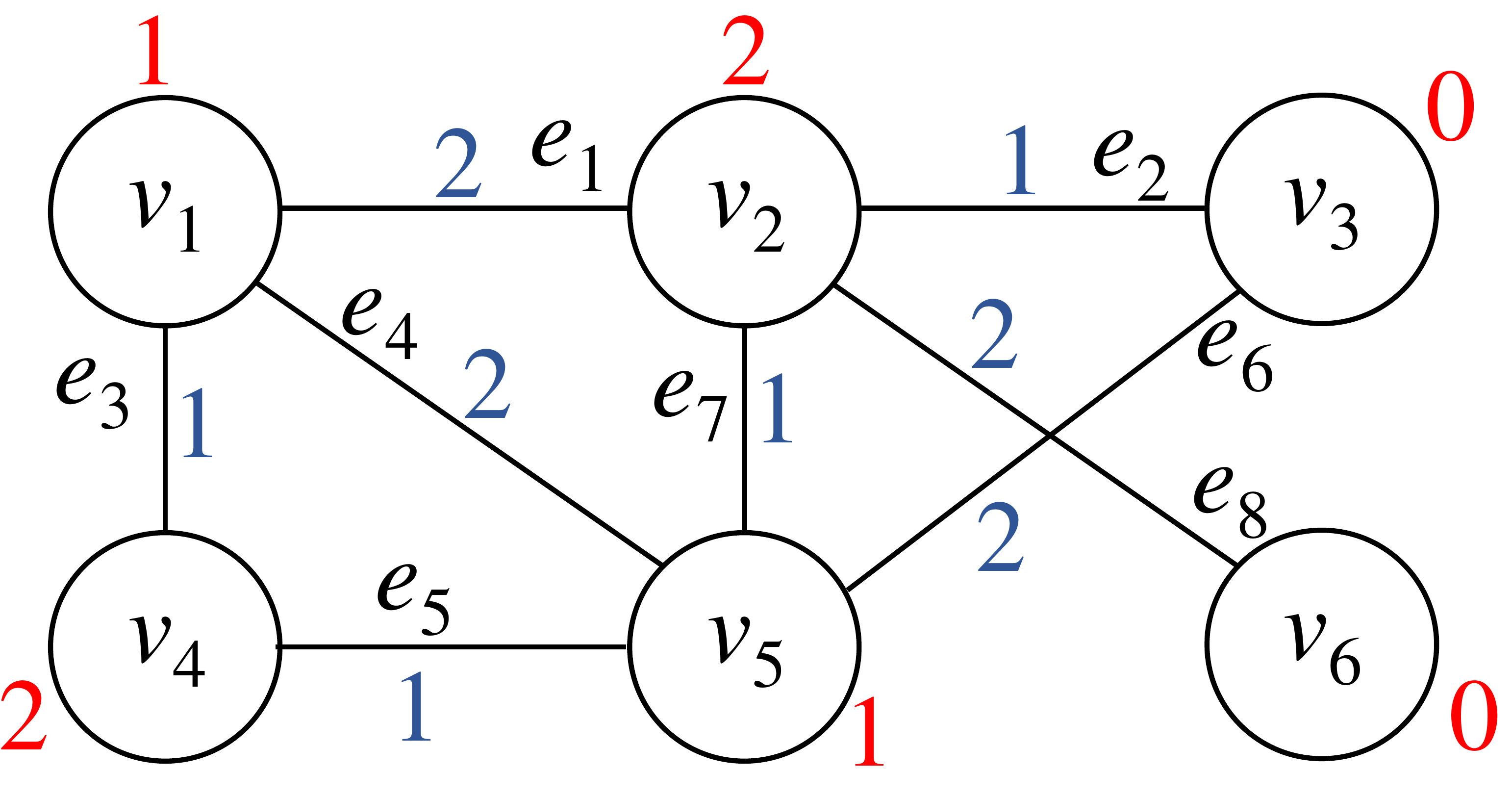}\caption{An example of a graph-temporal trajectory on an undirected graph, with the red numbers indicating node labels, and the blue numbers indicating edge labels, all at a fixed time index $k$.}
	\label{graph_exp}
\end{figure}

\begin{example}
	For the graph in Fig. \ref{graph_exp}, the node and edge labels are from a graph-temporal trajectory $g$ by fixing a time index $k$. The atomic node proposition $\pi=(x\le0)$
	is satisfied by $g$ at $v_{3}$ and $v_{6}$ at time index $k$. The edge proposition $\rho=(y\ge2)$
	is satisfied by $g$ at $e_1$, $e_{4}$, $e_{6}$ and $e_{8}$ at time index $k$.
\end{example}

\begin{definition}
	\label{neighbor}
	Given a graph-temporal trajectory $g=(x,y)$ on a graph $G$, a subset $V'\subseteq V$ of nodes and an edge proposition $\rho$, we define the \textit{neighbor operation} $\bigcirc_{\rho}$ : $2^{V}\times\mathbb{T}\rightarrow 2^{V}\times\mathbb{T}$ as
	\[
	\begin{split}
	&\bigcirc_{\rho}(V',k) =\\&\big(\{v | \exists v'\in V', \exists e\in E, s(e)=\{v',v\}, (g,e,k)\models\rho\},k\big).
	\end{split}
	\]
	Intuitively, $\bigcirc_{\rho}(V',k)$ consists of nodes that can be reached from $V'$ through an edge where the edge proposition $\rho$ is satisfied by $g$ at time index $k$. Note that neighbor operations can be applied successively. 	
\end{definition}

\begin{example}
	For the graph-temporal trajectory $g$ on the graph $G$ at time index $k$ in Fig. \ref{graph_exp},
	\[
	\begin{split}
	\bigcirc_{y\le1}(\{v_{4}\},k)&=\big(\{v_{1}, v_{5}\},k\big), \\ \bigcirc_{y\le1}\bigcirc_{y\le1}(\{v_{4}\},k)&=\bigcirc_{y\le1}\big(\{v_{1}, v_{5}\},k\big)=\big(\{v_{2}, v_{4}\},k\big).                                   
	\end{split}
	\]
\end{example}
\vskip 0.05in

\subsection{pGTL Formulas and GTL formulas}
We define the syntax of a parametric graph temporal logic (pGTL) formula $\varphi$ recursively as    
\begin{center}
	$\varphi:=\pi~|~\exists^{N}(\bigcirc_{\rho_{n}}\cdots \bigcirc_{\rho_{1}})\varphi~|~\neg\varphi~|~\varphi\wedge\varphi~|~\varphi\mathcal{U}\varphi~|~\Diamond_{\sim i}\varphi$, 
\end{center}
where $n$ and $N$ are positive integers, $\pi$ is an atomic         
node proposition, $\rho_i~(i=1,\dots,n)$ are edge propositions, $\exists^{N}(\bigcirc_{\rho_{n}}\cdots \bigcirc_{\rho_{1}})\varphi$ reads as \textquotedblleft there exists at least $N$ nodes under the neighbor operation $\bigcirc_{\rho_{n}}\cdots \bigcirc_{\rho_{1}}$ that satisfy $\varphi$ \textquotedblright, $\lnot$ and $\wedge$ stand for negation and conjunction respectively, $\mathcal{U}$ is a temporal operator representing \textquotedblleft until\textquotedblright, $\Diamond_{\sim i}$ is a parametrized temporal operator representing \textquotedblleft parametrized eventually\textquotedblright, where $\sim\in\{\ge,\le\}$ and $i\in\mathbb{T}$ is a temporal parameter. We can also derive $\vee$ (disjunction), $\Diamond$ (eventually), $\Box$ (always), $\Box_{\sim i}$ (parametrized always), $\mathcal{U}_{\sim i}$ (parametrized until) and $\Rightarrow$ (implication) from the above-mentioned operators \cite{pLTL2014}. We can also derive the parametrized temporal operators $\Diamond_{\ge i_1, \le i_2}$ and $\Box_{\ge i_1, \le i_2}$ ($i_1<i_2, i_1, i_2\in\mathbb{T})$ as
$\Diamond_{\ge i_1, \le i_2}\phi=\Diamond_{\ge i_1}\phi\wedge\Diamond_{\le i_2}\phi$ and
$\Box_{\ge i_1, \le i_2}\phi=\Box_{\ge i_1}\phi\wedge\Box_{\le i_2}\phi$.

The satisfaction relation $(g, v, k)\models\varphi$ for a graph-temporal trajectory $g$ at node $v$ at time index $k$ with respect to a pGTL formula $\varphi$ is defined recursively as follows:                                             
\[
\begin{split}
(g, v, k)\models\pi\quad\mbox{iff}\quad& x(v, k)\in\mathcal{O}(\pi)\\
(g, v, k)\models\exists^{N}(\bigcirc_{\rho_{n}}\cdots \bigcirc_{\rho_{1}})\varphi\quad\mbox{iff}\quad& \exists
v_1, \dots, v_N~(v_i\neq v_j ~\mbox{for} i\neq j), \mbox{s.t.} \\& \forall i, (v_i,k)\in\bigcirc_{\rho_{n}}\cdots \bigcirc_{\rho_{1}}(v,k)~ \mbox{and}~ (g, v_i, k)\models\varphi   \\             
(g, v, k)\models\lnot\varphi\quad\mbox{iff}\quad & (g, v, k)\not\models\varphi\\
(g, v, k)\models\varphi_{1}\wedge\varphi_{2}\quad\mbox{iff}\quad & (g, v, k)\models\varphi_{1}~\mbox{and}~(g, v, k)\models\varphi_{2}\\
(g, v, k)\models\varphi_{1}\mathcal{U}\varphi_{2}\quad\mbox{iff}\quad &  \exists
k'\ge k, \mbox{s.t.}~(g, v, k')\models\varphi_{2},\\
&  (g, v, k^{\prime\prime})\models\varphi_{1}, \forall k^{\prime\prime}\in[k, k']\\
(g, v, k)\models\Diamond_{\sim i}\varphi\quad\mbox{iff}\quad & \exists
k'\sim k+i, \mbox{s.t.}~(g, v, k')\models\varphi
\end{split}
\]
Intuitively, $\exists^{N}(\bigcirc_{\rho_{n}}\cdots \bigcirc_{\rho_{1}})\varphi$ is satisfied by a graph-temporal trajectory $g$ at a node $v\in V$ and at a time index $k$ if there exist at least $N$ nodes in $(\bigcirc_{\rho_{n}}\cdots \bigcirc_{\rho_{1}})(v,k)$ where $\varphi$ is satisfied by $g$ at time index $k$. Note that, by definition, if $(\bigcirc_{\rho_{n}}\cdots \bigcirc_{\rho_{1}})(v,k)$ consists of fewer than $N$ nodes, then $\exists^{N}(\bigcirc_{\rho_{n}}\cdots \bigcirc_{\rho_{1}})\varphi$ is false. 

We also define that a graph-temporal trajectory $g$ satisfies a pGTL formula $\varphi$ at a node $v$, denoted as $(g, v)\models\varphi$, if $g$ satisfies $\varphi$ at node $v$ at time index 1.

\begin{definition}
	We define a graph temporal logic (GTL) formula $\varphi_{\theta}$ as a pGTL formula $\varphi$ with fixed parameter valuation $\theta$.
\end{definition}

\begin{example}
	For the pGTL formula $\varphi=\exists^{N}\bigcirc_{y\le a} (x\ge b)$, we can induce a GTL formula $\varphi_{\theta}=\exists^{2}\bigcirc_{y\leq 1} (x\geq1)$ with $\theta([N,a,b])=[2,1,1]$. For the graph-temporal trajectory $g$ on the graph shown in Fig. \ref{graph_exp}, the set of nodes where $\varphi_{\theta}$ is satisfied by $g$ at time index $k$ are $\{v_{4}, v_{5}\}$. 
\end{example}

\subsection{Subtypes of pGTL and GTL Formulas}
Graph-temporal trajectories of finite time length are sufficient to satisfy (resp. violate) \textit{syntactically co-safe} (resp. \textit{safe}) pGTL formulas, which are defined as follows.

\begin{definition}
	The syntax of the \textit{syntactically co-safe} pGTL formula is defined recursively as
	\[
	\begin{split}
	\varphi:=&\top\mid\pi\mid\lnot\pi\mid\varphi_{1}\wedge\varphi_{2}\mid\varphi_{1}\vee
	\varphi_{2}\mid \Diamond\varphi\mid \varphi_1\mathcal{U}\varphi_2 \mid \Diamond_{\sim i}\varphi\\&\mid \Box_{\le i}\varphi\mid \varphi_1\mathcal{U}_{\sim i}\varphi_2~|~\exists^{N}(\bigcirc_{\rho_{n}}\cdots \bigcirc_{\rho_{1}})\varphi.
	\end{split}
	\]  
\end{definition}  
\vskip 0.05in

\begin{definition}                                                                                                                             
	The syntax of the \textit{syntactically safe} pGTL formula is defined as 
	\[
	\begin{split}
	\varphi:=&\bot\mid\pi\mid\lnot\pi\mid\varphi_{1}\wedge\varphi_{2}\mid\varphi_{1}\vee
	\varphi_{2}\mid \Box\varphi \mid \Diamond_{\le i}\varphi\mid \Box_{\sim i}\varphi\\& \mid \varphi_1\mathcal{U}_{\le i}\varphi_2~|~\exists^{N}(\bigcirc_{\rho_{n}}\cdots \bigcirc_{\rho_{1}})\varphi.
	\end{split}
	\]
\end{definition}             
\vskip 0.05in

We further introduce type-I  and type-II pGTL formulas.

(1) A type-I pGTL formula $\varphi$ is defined as
\[
\varphi:=\exists^{N}(\bigcirc_{\rho_{n}}\cdots \bigcirc_{\rho_{1}})\pi~|~\neg\varphi~|~\varphi\wedge\varphi~|~\varphi\mathcal{U}\varphi~|~\Diamond_{\sim i}\varphi.
\]

(2) A type-II pGTL formula $\varphi$ is defined recursively as
\begin{align}\nonumber
\varphi:=\exists^{N}(\bigcirc_{\rho_{n}}\cdots \bigcirc_{\rho_{1}})\phi,          
\label{template1}
\end{align}                                       
where $\phi$ is defined recursively as
\[
\phi:=\pi~|~\neg\phi~|~\phi\wedge\phi~|~\phi\mathcal{U}\phi~|~\Diamond_{\sim i}\phi.                                               
\]
 The subtypes of GTL Formulas can be defined similarly with fixed parameter valuations. A pGTL formula can be neither a type-I pGTL formula nor a type-II pGTL formula. A pGTL formula could also be of more than one subtypes. For example, if a pGTL formula is both a type-I pGTL formula and a syntactically co-safe pGTL formula, then it is called a type-I syntactically co-safe pGTL formula. 

\begin{definition} 
	A deterministic finite automaton (DFA) is a tuple $\mathcal{A}=(\mathcal{Q},q_{0},\Sigma,\delta,Acc)$ where
	$\mathcal{Q}=\{q_0, q_1, \dots, q_{K-1}\}$ is a finite set of states, $q_{0}$ is the initial state, $\Sigma$ is the alphabet, $\delta:\mathcal{Q}\times\Sigma\rightarrow\mathcal{Q}$ is the transition relation and $Acc\subseteq2^{\mathcal{Q}}$ is a finite set of accepting states \cite{KupfermanVardi2001}. 
	\label{sw}                               
\end{definition}                                                                          

We use $AP^{\textrm{I}}$ and $AP^{\textrm{II}}$ to denote the sets of \textit{atomic predicates} in the form of $\exists^{N}(\bigcirc_{\rho_{n}}\cdots \bigcirc_{\rho_{1}})\pi$ and $\pi$, respectively. At a node $v$ and over $AP^{\textrm{I}}$ (resp. $AP^{\textrm{II}}$), the \textit{word} generated by a graph-temporal trajectory $g=(x,y)$ is a sequence $\mathcal{L}^v\big(x(\cdot,1),y(\cdot,1)\big),\dots,\mathcal{L}^v\big(x(\cdot,L),y(\cdot,L)\big)$, where $\mathcal{L}^v: \mathcal{X}\times\mathcal{Y}\rightarrow2^{\mathcal{AP}^{\textrm{I}}}$ (resp. $2^{\mathcal{AP}^{\textrm{II}}}$) is a labeling function assigning a subset of atomic predicates in $\mathcal{AP}^{\textrm{I}}$ (resp. $\mathcal{AP}^{\textrm{II}}$) that hold true at node $v$ to each $\big(x(\cdot,k),y(\cdot,k)\big)$, $k\in[1,L]$.         
If a pGTL formula $\varphi$ is a type-I pGTL formula and it is syntactically co-safe (resp. safe), then we can build a DFA $\mathcal{A}^{\varphi_{\theta},v}$ (resp. $\mathcal{A}^{\lnot\varphi_{\theta},v}$) over $AP^{\textrm{I}}$ that accepts precisely the words generated by graph-temporal trajectories that satisfy (resp. violate) the GTL formula $\varphi_{\theta}$ at node $v$ for any $\theta$. If a pGTL formula $\varphi=\exists^{N}(\bigcirc_{\rho_{n}}\cdots \bigcirc_{\rho_{1}})\phi$ is a type-II pGTL formula and it is syntactically co-safe (resp. safe), we can build a DFA $\mathcal{A}^{\phi_{\theta},v}$  (resp. $\mathcal{A}^{\lnot\phi_{\theta},v}$) over $AP^{\textrm{II}}$ that accepts precisely the words generated by graph-temporal trajectories that satisfy (resp. violate) $\phi_{\theta'}$ at node $v$ for any $\theta'$ \cite{KupfermanVardi2001}.

In the following sections of the paper, we only focus on type I/type II syntactically co-safe/syntactically safe pGTL (GTL) formulas and we simply call them pGTL (GTL) formulas for conciseness. We define that a syntactically co-safe GTL formula $\varphi_{\theta}$ is violated by a graph-temporal trajectory $g$ of finite time length $L$ at a node $v$ if $\varphi_{\theta}$ is not satisfied by $g$ at node $v$;  and a syntactically safe GTL formula $\varphi_{\theta}$ is satisfied by a graph-temporal trajectory $g$ of finite time length $L$ at a node $v$ if $\varphi_{\theta}$ is not violated by $g$ at node $v$.

\section{Graph Temporal Logic Inference for Classification}
\label{classify_sec}
In this section, we present the problem formulation and solution to infer GTL formulas for classification. 

Suppose that we are given a set $\mathcal{D}^G_L=\{(g_k,l_k)\}^{N_{\mathcal{D}}}_{k=1}$ of \textit{labeled} graph-temporal trajectories of time length $L$ on a graph $G$, where the \textit{classification labels} $l_k=1$ and $l_k=-1$ represent desired and undesired behaviors, respectively. For a GTL formula $\varphi_{\theta}$, we define the \textit{satisfaction signature} $\zeta_{\varphi_{\theta}}(g_k,v)$ of a graph-temporal trajectory $g_k$ at node $v$ as follows: $\zeta_{\varphi_{\theta}}(g_k,v)=1$ if $(g_k, v)\models\varphi_{\theta}$; and $\zeta_{\varphi_{\theta}}(g_k,v)=-1$ if $(g_k, v)\not\models\varphi_{\theta}$. A labeled graph-temporal trajectory $(g_k, l_k)$ is misclassified by $\varphi_{\theta}$ at node $v$ if $\zeta_{\varphi_{\theta}}(g_k,v)\neq l_k$. 

We define the nodal misclassification rate of $\varphi_{\theta}$ in $\mathcal{D}^G_L$ as
\[
MR(\mathcal{D}^G_L, \varphi_{\theta})=\frac{\sum_{v\in V}\vert\{(g_k,l_k)\in\mathcal{D}^G_L: \zeta_{\varphi_{\theta}}(g_k,v)\neq l_k\}\vert}{\vert\mathcal{D}^G_L\vert\vert V\vert},
\]
where $\vert S \vert$ denotes the cardinality of a set $S$. 

The \textit{size} of a GTL formula $\varphi_{\theta}$, denoted as $\eta(\varphi_{\theta})$, is defined as the number of \textit{Boolean connectives} (i.e., conjunctions or disjunctions) in $\varphi_{\theta}$. Note that logically equivalent formulas may have different sizes.         

\begin{problem} 
	Given a dataset $\mathcal{D}^G_L=\{g_1, \dots,$ $g_m\}$, a prior probability distribution $\mathcal{F}^G_L$, a real constant $m_{\rm{th}}\in[0,1)$ and an integer constant $\eta_{\rm{th}}\in(0, \infty)$, construct a GTL formula $\varphi_{\theta}$ that satisfies the following two constraints:
	\begin{itemize}
		\item \textit{classification constraint}: $MR(\mathcal{D}^G_L,\varphi_{\theta})\le m_{\rm{th}}$, i.e.,
		the nodal misclassification rate should not exceed $m_{\rm{th}}$;
		\item \textit{size constraint}:
		$\eta(\varphi_{\theta})\le\eta_{\rm{th}}$, i.e., the size of $\varphi_{\theta}$ should not exceed $\eta_{\rm{th}}$. 
	\end{itemize}       
	\label{problem1}
\end{problem}

Most existing approaches for inferring temporal logic formulas for classification apply readily to solve Problem \ref{problem1}. As an example, we use the \textit{pruning and growing} approach illustrated in \cite{Kong2017}. We start from a set $\mathcal{P}$ of \textit{primitive} pGTL formulas (also called \textit{templates}), i.e., pGTL formulas that do not contain conjunctions or disjunctions. We use particle swarm optimization (PSO) \cite{PSO_Eberhart} to compute the parameter valuation $\theta$ for each primitive pGTL formula from $\mathcal{P}$ that minimizes the nodal misclassification rate $MR(\mathcal{D}^G_L,\varphi_{\theta})$. Other global optimization methods such as simulated annealing \cite{Kong2017} and Monte-carlo sampling \cite{Nghiem2010} are also valid candidates for such computations. If a GTL formula that satisfies the classification constraint is not found, we only keep the pGTL formulas in $\mathcal{P}$ such that the misclassification rates can be achieved below a threshold $\hat{m}_{\rm{th}}\in(m_{\rm{th}},1)$ (\textit{pruning}). Then we infer a GTL formula $\varphi_{\theta}$ in the form of $\varphi_{\theta}=\varphi^1_{\theta_1}\vee\varphi^2_{\theta_2}$ or $\varphi_{\theta}=\varphi^1_{\theta_1}\wedge\varphi^2_{\theta_2}$ (\textit{growing}), where $\varphi^1$ and $\varphi^2$ are chosen from the pGTL formulas kept in the first step. In this way, we keep increasing the number of primitive pGTL formulas connected with conjunctions or disjunctions until a GTL formula that satisfies the classification constraint is found, or the size constraint is violated.

\section{Graph Temporal Logic Inference for Identification}
\label{info_sec}
In this section, we present the problem formulation and solution to identify GTL formulas from data. 

\subsection{Problem Formulation}      
We use $\mathcal{B}^G_L$ to denote the set of all possible graph-temporal trajectories with time-length $L$ on the graph $G$. We are given a dataset $\mathcal{S}^G_L=\{g_1, \dots, g_m\}\subset\mathcal{B}^G_L$ as a collection of graph-temporal trajectories. We use $\mathcal{F}^G_L: \mathcal{B}^G_L\rightarrow[0,1]$ to denote a prior probability distribution over $\mathcal{B}^G_L$, and $\mathbb{P}_{\mathcal{F}^G_L,\varphi_{\theta},v}$ to denote the probability of a GTL formula $\varphi_{\theta}$ being satisfied at node $v$ based on $\mathcal{F}^G_L$, i.e.,
\[
\mathbb{P}_{\mathcal{F}^G_L,\varphi_{\theta},v}=\mathbb{P}\{(g,v)\models\varphi_{\theta}\}, g\sim\mathcal{F}^G_L.
\]

\begin{assumption}
	\label{assume}
	We assume that every graph-temporal trajectory in $\mathcal{B}^G_L$ occurs with non-zero probability based on $\mathcal{F}^G_L$. 
\end{assumption}

From Assumption \ref{assume}, for any $g\in\mathcal{B}^G_L$, if $(g, v)\models\varphi_{\theta}$, then $\mathbb{P}_{\mathcal{F}^G_L,\varphi_{\theta},v}>0$. 

\begin{definition}
	Given a prior probability distribution $\mathcal{F}^G_L$, we define $\mathcal{\bar{F}}^{\varphi_{\theta},v}_L: \mathcal{B}^G_L\rightarrow[0,1]$ as the posterior probability distribution given that the GTL formula $\varphi_{\theta}$ evaluates to true at node $v$, which is expressed as
	\[                           
	\begin{split}
	\mathcal{\bar{F}}^{\varphi_{\theta},v}_L(g)=                                    
	\begin{cases}
	\frac{\mathcal{F}^G_L(g)}{\mathbb{P}_{\mathcal{F}^G_L,\varphi_{\theta},v}} ~~~~~~~~~~~\mbox{if}~(g,v)\models\varphi_{\theta},\\ 
	0~~~~~~~~~~~~~~~~~~~~\mbox{if}~(g,v)\not\models\varphi_{\theta}.
	\end{cases} 
	\end{split}
	\]
	\label{know}
\end{definition}  
\begin{remark}
	The expression of $\mathcal{\bar{F}}^{\varphi_{\theta},v}_L$ can be directly derived using Bayes' theorem.                             
\end{remark}

\begin{definition}
	We define 
	\[                           
	\begin{split}
	\mathcal{I}(\mathcal{F}^G_L,\mathcal{\bar{F}}^{\varphi_{\theta},v}_L):=\frac{1}{L}\cdot D_{\rm{KL}}(\mathcal{\bar{F}}^{\varphi_{\theta},v}_L\vert\vert\mathcal{F}^G_L)
	\end{split}
	\]
	as the information gain when the prior probability distribution $\mathcal{F}^G_L$ is updated to the posterior probability distribution $\mathcal{\bar{F}}^{\varphi_{\theta},v}_L$, where $D_{\rm{KL}}(\mathcal{\bar{F}}^{\varphi_{\theta},v}_L\vert\vert\mathcal{F}^G_L)$
	is the Kullback-Leibler divergence from $\mathcal{F}^G_L$ to $\mathcal{\bar{F}}^{\varphi_{\theta},v}_L$.                       
	\label{KL}
\end{definition}   

\begin{remark}
	If $\varphi_{\theta}=\top$, then obviously $\mathbb{P}_{\mathcal{F}^G_L,\varphi_{\theta},v}=1$ and $\mathcal{I}(\mathcal{F}^G_L,\mathcal{\bar{F}}^{\varphi_{\theta},v}_L)=0$ for any node $v$, i.e., tautologies provide no information gain. For completeness, we also define that the information gain $\mathcal{I}(\mathcal{F}^G_L,$ $\mathcal{\bar{F}}^{\varphi_{\theta},v}_L)=0$ for any node $v$ if $\mathbb{P}_{\mathcal{F}^G_L,\varphi_{\theta},v}=0$. So if $\varphi_{\theta}=\bot$, then $\mathbb{P}_{\mathcal{F}^G_L,\varphi_{\theta},v}=0$ and $\mathcal{I}(\mathcal{F}^G_L,\mathcal{\bar{F}}^{\varphi_{\theta},v}_L)=0$ for any node $v$, i.e., contradictions provide no information gain.   
\end{remark} 

\begin{definition}
	We define $\chi(\varphi_{\theta}, \mathcal{S}^G_L)$ as the averaged proportion of nodes (in $G$) at which the GTL formula $\varphi_{\theta}$ is satisfied in the dataset $\mathcal{S}^G_L$, i.e.,
	\[
	\chi(\varphi_{\theta}, \mathcal{S}^G_L)=\frac{\sum_{i}\vert\{v\in V~\vert~(g_i,v)\models\varphi_{\theta}\}\vert}{\vert V\vert\vert\mathcal{S}^G_L\vert}.
	\]
\end{definition}             

\begin{problem} 
	Given a dataset $\mathcal{S}^G_L$, a prior probability distribution $\mathcal{F}^G_L$ and a real constant $p_{\rm{th}}\in(0,1]$, compute the parameter valuation $\theta$ for a pGTL formula $\varphi$ (selected from a set $\mathcal{P}$ of templates) that maximizes the average information gain at each node $\frac{1}{\vert V\vert}\sum_{v\in V}\mathcal{I}(\mathcal{F}^G_L,\mathcal{\bar{F}}^{\varphi_{\theta},v}_L)$ while satisfying the \textit{coverage constraint}: $\chi(\varphi_{\theta}, \mathcal{S}^G_L)\ge p_{\rm{th}}$, i.e.,
    $\varphi_{\theta}$ is satisfied for at least $p_{\rm{th}}$ proportion of nodes in average in $\mathcal{S}^G_L$.   
	\label{problem_info}                                                                                                                        
\end{problem}

Note that we focus on the parameter identification problem based on pGTL formulas selected from a set of templates and thus the size constraint is not needed. 

To solve Problem \ref{problem_info}, it is computationally inefficient to use optimzation algorithms such as PSO and compute the average information gain for each candidate GTL formula (see Sec.  \ref{sec_compute_info} for the time complexity). In Sec.  \ref{minimalSet}, we prove that the optimal parameter valuation of a GTL formula lies in the
\textit{minimal satisfying set} of parameter valuations. Thus we only need to compute the average information gain for the GTL formulas with parameter valuation in the (approximated) minimal satisfying set (see Sec. \ref{identify_pGTL}).

\subsection{Computation of Information Gain of GTL Formulas}      
\label{sec_compute_info}
In this subsection, we present the algorithm for computing the average information gain for GTL formulas. 

\begin{proposition} 
	For a GTL formula $\varphi_{\theta}$, if $\mathbb{P}_{\mathcal{F}^G_L,\varphi_{\theta},v}>0$,   
	\begin{center}                 
	$\mathcal{I}(\mathcal{F}^G_L,\mathcal{\bar{F}}^{\varphi_{\theta},v}_L) 
	=\frac{-\log{\mathbb{P}_{\mathcal{F}^G_L,\varphi_{\theta},v}}}{L}$.
	\end{center}    
	\label{exp}   
\end{proposition}  
\vskip 0.1in
\begin{proof} 
	Straightforward from Definitions \ref{know} and \ref{KL}.
\end{proof} 

In the following, for computational efficiency we choose prior probability distributions such that there exist no spatial or temporal dependencies for the labels on different nodes and edges. Note that the proposed methodology readily applies to cases where such spatial or temporal dependencies exist, e.g., when the prior distribution is governed by Markov random fields (spatial dependence) or discrete-time Markov chains (temporal dependence), but the computational complexity is significantly increased (see Sec. IV-B of \cite{zhe_info} for an example for parametric linear temporal logic).

Algorithm \ref{compute} is for computing the avearge information gain for GTL formulas. We first explain how to compute the avearge information gain for a type-I syntactically co-safe GTL formula $\varphi_{\theta}$. 
We use $p_L^{\varphi_{\theta}, v}(\ell,q_{k})$ to denote the probability of a graph-temporal trajectory of time length $L$ satisfying $\varphi_{\theta}$ at node $v$, conditioned on the fact that the state of the DFA $\mathcal{A}^{\varphi_{\theta}, v}=(\mathcal{Q}^{\varphi_{\theta}, v},q^{\varphi_{\theta}, v}_{0},2^{AP^{\textrm{I}}},\delta^{\varphi_{\theta}, v},Acc^{\varphi_{\theta}, v})$ at time index $\ell$ ($1\le\ell\le L$) being the state $q_{k}$. We first initialize $p_L^{\varphi_{\theta}, v}(L,q_{k})$ as (Line \ref{init})                                                                     
\[
\begin{split}
p_L^{\varphi_{\theta}, v}(L,q_{k})=                                                            
\begin{cases}
1  ~~~~~~~~~~~~~~~\mbox{if}~q_k\in Acc^{\varphi_{\theta}, v};\\ 
0 ~~~~~~~~~~~~~~~~\mbox{otherwise}.
\end{cases}
\end{split}
\]
We can compute $p_L^{\varphi_{\theta}, v}(\ell,q_{k})$ recursively as (Line \ref{iterate})
\begin{align}\nonumber
\begin{split}
\begin{bmatrix}
p_L^{\varphi_{\theta}, v}(\ell-1,q_{0})\\
\vdots\\
p_L^{\varphi_{\theta}, v}(\ell-1,q_{K})\\	
\end{bmatrix}=
\begin{bmatrix}
c^{\ell}_{0,0}  & \dots & c^{\ell}_{0,K}\\                                                                                                
\vdots   & \vdots & \vdots \\
c^{\ell}_{K,0}  & \dots & c^{\ell}_{K,K}\\
\end{bmatrix}
\begin{bmatrix}
p_L^{\varphi_{\theta}, v}(\ell,q_{0})\\
\vdots\\
p_L^{\varphi_{\theta}, v}(\ell,q_{K})\\
\end{bmatrix},	
\end{split}                        
\end{align}
where $c^{\ell}_{j,k}$ is the probability of transitioning from $q_{j}$ to $q_{k}$ at time index $\ell$ and $c^{\ell}_{j,k}$ can be calculated based on $\mathcal{F}^G_L$ (Line \ref{calculate}).
Finally, we have $\mathbb{P}_{\mathcal{F}^G_L,\varphi_{\theta}, v}=p_L^{\varphi_{\theta}, v}(1,q_{0})$ (Line \ref{cosafe}). 

For type-I syntactically safe GTL formulas, by replacing each $\varphi_{\theta}$ in the above deductions with $\lnot\varphi_{\theta}$, we can compute $\mathbb{P}_{\mathcal{F}^G_L,\varphi_{\theta},v}=1-p_L^{\lnot\varphi_{\theta},v}(1,q_{0})$ (Line \ref{safe}). 

For a type-II GTL formula $\varphi_{\theta}=\exists^{N}(\bigcirc_{\rho_{n}}\cdots\bigcirc_{\rho_{1}})\phi_{\theta'}$, $\mathbb{P}_{\mathcal{F}^G_L,\varphi_{\theta},v}$ can be computed as (Line \ref{gamma_line})
\[                           
\begin{split}
\sum_{k=N}^{N^v_{\varphi_{\theta}}}\frac{N^v_{\varphi_{\theta}}!}{k!(N^v_{\varphi_{\theta}}-k)!}\mathbb{P}_{\mathcal{F}^G_L,\phi_{\theta'},v}^k(1-\mathbb{P}_{\mathcal{F}^G_L,\phi_{\theta'},v})^{N^v_{\varphi_{\theta}}-k},
\end{split}
\]	
where $N^v_{\varphi_{\theta}}=\vert\bigcirc_{\rho_{n}}\cdots\bigcirc_{\rho_{1}}(v,1)\vert$ is the number of nodes in the set $\bigcirc_{\rho_{n}}\cdots\bigcirc_{\rho_{1}}(v,1)$, and $\mathbb{P}_{\mathcal{F}^G_L,\phi_{\theta'},v}$ can be computed in a similar way for type-I GTL formulas.

With $\mathbb{P}_{\mathcal{F}^G_L,\varphi_{\theta},v}$, we can compute $\mathcal{I}(\mathcal{F}^G_L,\mathcal{\bar{F}}^{\varphi_{\theta},v}_L)$ according to Proposition \ref{exp} (Line \ref{last_line}).  

The time complexity of Algorithm \ref{compute} is $O(\vert V\vert LK^2)$, where $\vert V\vert$ is the number of nodes in the graph, $L$ is the time length of graph-temporal trajectories, $K$ is the number of states of $\mathcal{A}^{\varphi_{\theta},v}$ (for type-I formulas) or $\mathcal{A}^{\phi_{\theta'},v}$ (for type-II formulas).

\begin{algorithm}
	\DontPrintSemicolon
	\SetKwBlock{Begin}{function}{end function}       
	\Begin($\text{ComputeIG} {(} \mathcal{S}^G_L=\{g_1, \dots,g_m\},\varphi_{\theta},\mathcal{F}^G_L {)}$)
	{
		\If{$\varphi_{\theta}$ is a type-I GTL formula}
		{
		 $\psi\gets\varphi_{\theta}$
		} 
		\ElseIf{$\varphi_{\theta}$ is a type-II GTL formula}
		{ For $\varphi_{\theta}=\exists^{N}(\bigcirc_{\rho_{n}}\cdots \bigcirc_{\rho_{1}})\phi_{\theta'}$, $\psi\gets\phi_{\theta'}$ }  
		\For{$v\in V$}
		{Obtain the DFA $\mathcal{A}^{\psi,v}$ (resp. $\mathcal{A}^{\lnot\psi,v}$) if $\psi$ is syntactically co-safe (resp. syntactically safe)\;
	\For{$k=0$ to $K$}  
	{ Initialize $p_L^{\psi, v}(L,q_{k})$ \label{init} }	     
	\For{$\ell=L$ to 2, $j=0$ to $K$}  
	{ For each $k\in[0,K]$, calculate $c^{\ell}_{j,k}$ \label{calculate}                                                                                                       
		$p_L^{\psi, v}(\ell-1,q_{j})\gets\sum_{k=0}^{K}c^{\ell}_{j,k}p_L^{\psi, v}(\ell,q_{k})$ \label{iterate}                }                                                                  
	$\beta\gets p_L^{\psi, v}(1,q_{0})$ if $\psi$ is syntactically co-safe\; \label{cosafe}
	$\beta\gets 1-p_L^{\psi, v}(1,q_{0})$ if $\psi$ is syntactically safe\; \label{safe}  
	\If{$\psi$ is a type-I GTL formula}                                              
	{$\gamma_v\gets\beta$}
	\ElseIf{$\psi$ is a type-II GTL formula}                                                     
	{$\gamma_v\gets\displaystyle\sum_{k=N}^{N^v_{\psi}}\frac{N^v_{\psi}!}{k!(N^v_{\psi}-k)!}\beta^k(1-\beta)^{N^v_{\psi}-k}$ }    \label{gamma_line}       }                                                                       
\Return{$\mathcal{I}=-\frac{1}{\vert V\vert L}\sum_{v\in V}\log(\gamma_v)$} \label{last_line}
}
		\caption{Avearge information gain computation for GTL formulas.}                                           
		\label{compute}
\end{algorithm}	
%\vspace{-30pt}

\subsection{Minimal Satisfying Set of Parameter Valuations}
\label{minimalSet}

In this subsection, we introduce some related definitions and lemmas, leading to the result in Proposition \ref{pareto}.

\begin{definition}
 The polarity $\varrho(p, \varphi)$ of a scalar parameter $p$ for a pGTL formula $\varphi$ is defined recursively as
\begin{gather*}\nonumber
\varrho(p,\neg\varphi)=\sim\varrho(p,\varphi),~\varrho(p, f(x)\le p)=+,~ \varrho(p, f(x)\ge p)=-,\\
\varrho(p,\Diamond_{\le p}\varphi)=+\circ\varrho(p, \varphi), ~~\varrho(p,\Diamond_{\ge p}\varphi)=-\circ\varrho(p, \varphi),\\
\varrho(p,\varphi \mathcal{U}\psi)=\varrho(p, \varphi\wedge\psi)=\varrho(p, \varphi)\circ\varrho(p, \psi),\\
\varrho(p,\exists^{N}\bigcirc_{y\le p}\varphi)=+\circ\varrho(p, \varphi),\\
\varrho(p,\exists^{N}\bigcirc_{y\ge p}\varphi)=-\circ\varrho(p, \varphi),\\
\varrho(p,\exists^{p}\bigcirc_{\rho}\varphi)=-\circ\varrho(p, \rho)\circ\varrho(p, \varphi),\\
\varrho(p, \varphi)=\textrm{U}, \mbox{iff~$p$~does~not~appear~in~} \varphi,
\end{gather*}
where $f$ is some real-valued function, the operations $\sim$ and $\circ$ are as defined in the following table \cite{Asarin2012}
\begin{equation}\nonumber
\begin{array}{C|C}
&	$\sim$ \\
\hline
U & U \\
+ & - \\
- & + \\
M & M 
\end{array}\qquad
\begin{array}{C|C C C C C}
$\circ$ & U & + & - & M\\
\hline
U & U & + & - & M\\
+ & + & + & M & M\\
- & - & M & - & M\\
M & M & M & M & M
\end{array}
\end{equation}
In this table, $\textrm{U}, +, -$ and $\textrm{M}$ represent undefined, positive, negative and mixed polarities respectively.
\label{polarity}
\end{definition}

Intuitively, the polarity $\varrho(p, \varphi)$ of a scalar parameter $p$ for a pGTL formula $\varphi$ is positive, if the pGTL formula $\varphi$ is easier to be satisfied when $p$ is increased; and it is negative, if it is easier to be satisfied when $p$ is decreased.

\begin{definition}
	For a graph $G$, we say that $\theta$ \textit{dominates} $\theta'$ with respect to a pGTL formula $\varphi$, denoted as $\theta\prec_{\varphi}\theta'$, if and only if $\varphi_{\theta}\Rightarrow\varphi_{\theta'}$ holds true and $\varphi_{\theta'}\Rightarrow\varphi_{\theta}$ does not hold true at any node $v$ for GTL formulas $\varphi_{\theta}$ and $\varphi_{\theta'}$.
	\label{dominate}
\end{definition} 

For example, for the six GTL formulas induced from the same pGTL formula $\varphi$ in Fig. \ref{chain}, we have $\varphi_{\theta_2}\prec\varphi_{\theta_1}$, $\varphi_{\theta_3}\prec\varphi_{\theta_2}$, $\varphi_{\theta_4}\prec\varphi_{\theta_2}$, $\varphi_{\theta_5}\prec\varphi_{\theta_2}$, $\varphi_{\theta_6}\prec\varphi_{\theta_3}$ and $\varphi_{\theta_6}\prec\varphi_{\theta_4}$.

\begin{figure}[th]
	\centering
	\includegraphics[width=12cm]{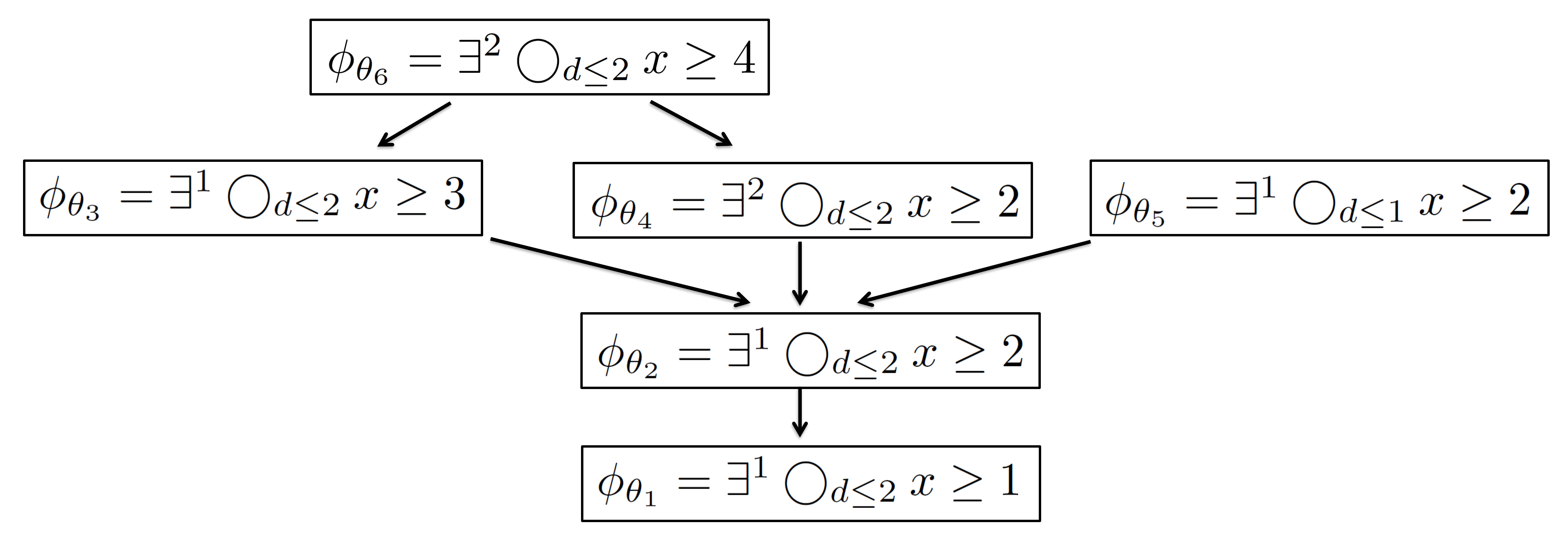}\caption{The \textit{domination} relationships among six GTL formulas.}
	\label{chain}
\end{figure} 

\begin{lemma}                                                                
	For two parameter valuations $\theta=[\theta_1,\dots,\theta_z]$ and $\theta'=[\theta'_1,\dots,\theta'_z]$, $\theta\prec_{\varphi}\theta'$ if $\theta\neq\theta'$ and the followings hold for each $i\in[1,z]$:
	\begin{align}\nonumber
	\begin{split}
	\begin{cases}
	\theta_i\le\theta'_i, ~~~~~&\mbox{if}~\varrho(\theta_i, \varphi) = +;\\
	\theta_i\ge\theta'_i, ~~~~~&\mbox{if}~\varrho(\theta_i, \varphi) = -.
	\end{cases}
	\end{split}
	\end{align} 
\end{lemma}  
 
\begin{proof}
	Straightforward from Definitions \ref{polarity} and \ref{dominate}.                                                  
\end{proof}

\begin{lemma}
	For two parameter valuations $\theta$ and $\theta'$, if $\theta\prec_{\varphi}\theta'$ for a pGTL formula $\varphi$, then $\mathcal{I}(\mathcal{F}^G_L,\mathcal{\bar{F}}^{\varphi_{\theta},v}_L)>\mathcal{I}(\mathcal{F}^G_L,\mathcal{\bar{F}}^{\varphi_{\theta'},v}_L)$ for any $v$ and $\mathcal{F}^G_L$. 
	\label{info_compare}
\end{lemma}                

\begin{proof}
	As $\theta\prec_{\varphi}\theta'$, we have $\varphi_{\theta}\Rightarrow\varphi_{\theta'}$ holds true and $\varphi_{\theta'}\Rightarrow\varphi_{\theta}$ does not hold true at any node $v$. Thus, for any node $v$ and any graph-temporal trajectory $g\in\mathcal{B}^G_L$, if $(g,v)\models\varphi_{\theta}$, then $(g,v)\models\varphi_{\theta'}$. Besides, for any node $v$ there exists at least one graph-temporal trajectory $\hat{g}\in\mathcal{B}^G_L$ such that $(\hat{g},v)\models\varphi_{\theta'}$ and $(\hat{g},v)\not\models\varphi_{\theta}$. From Assumption \ref{assume}, $\hat{g}$ occurs with non-zero probability based on $\mathcal{F}^G_L$. Therefore, $\mathbb{P}_{\mathcal{F}^G_L,\varphi_{\theta},v}<\mathbb{P}_{\mathcal{F}^G_L,\varphi_{\theta'},v}$. 
	Then according to Proposition \ref{exp}, $\mathcal{I}(\mathcal{F}^G_L,\mathcal{\bar{F}}^{\varphi_{\theta},v}_L)>\mathcal{I}(\mathcal{F}^G_L,\mathcal{\bar{F}}^{\varphi_{\theta'},v}_L)$ for any $v$ and $\mathcal{F}^G_L$.  
\end{proof}

For a pGTL formula $\varphi$, we denote by $\Theta_{sat}$ the set of parameter valuations $\theta$ such that $\varphi_{\theta}$ satisfies the coverage constraint. We further denote $\Theta_{unsat}:=\Theta\setminus\Theta_{sat}$.

\begin{definition}
	A parameter valuation $\theta$ in $\Theta_{sat}$ is said to be \textit{minimal} if there does not exist a parameter valuation $\theta'\in\Theta_{sat}$ such that $\theta'\prec_{\varphi}\theta$. We define the \textit{minimal satisfying set} $\Theta_s$ as the set of minimal parameter valuations in $\Theta_{sat}$.
	\label{pareto_optimal}
\end{definition} 

\begin{proposition} 
	For a prior distribution $\mathcal{F}^G_L$ and $\Theta_s$, if $\theta^{\ast}=\argmax\limits_{\theta\in\Theta_{sat}}\frac{1}{\vert V\vert}\sum_{v\in V}\mathcal{I}(\mathcal{F}^G_L,\mathcal{\bar{F}}^{\varphi_{\theta},v}_L)$, then $\theta^{\ast}\in\Theta_s$.
	\label{pareto}                                     
\end{proposition} 

\begin{proof}
	We prove Proposition \ref{pareto} by contradiction. If $\theta^{\ast}\not\in\Theta_s$, then from Definition \ref{pareto_optimal} there exists $\theta\in\Theta_{sat}$ such that $\theta\neq\theta^{\ast}$ and $\theta\prec_{\varphi}\theta^{\ast}$. Then from Lemma \ref{info_compare} we have $\mathcal{I}(\mathcal{F}^G_L,\mathcal{\bar{F}}^{\varphi_{\theta},v}_L)>\mathcal{I}(\mathcal{F}^G_L,\mathcal{\bar{F}}^{\varphi_{\theta^{\ast}},v}_L)$ for any node $v$. Thus, we have $\sum_{v\in V}\mathcal{I}(\mathcal{F}^G_L,\mathcal{\bar{F}}^{\varphi_{\theta},v}_L)>\sum_{v\in V}\mathcal{I}(\mathcal{F}^G_L,\mathcal{\bar{F}}^{\varphi_{\theta^{\ast}},v}_L)$. But as $\theta^{\ast}$ $=\argmax\limits_{\theta\in\Theta_{sat}}\frac{1}{\vert V\vert}\sum_{v\in V}\mathcal{I}(\mathcal{F}^G_L,\mathcal{\bar{F}}^{\varphi_{\theta},v}_L)$, we have $\sum_{v\in V}$ $\mathcal{I}(\mathcal{F}^G_L,\mathcal{\bar{F}}^{\varphi_{\theta},v}_L)\le\sum_{v\in V}\mathcal{I}(\mathcal{F}^G_L,\mathcal{\bar{F}}^{\varphi_{\theta^{\ast}},v}_L)$. Contradiction.
\end{proof}

From Proposition \ref{pareto}, it can be seen that the optimal parameter valuation belongs to the minimal satisfying set.

\subsection{Information-Guided Identification of GTL Formulas}    
\label{identify_pGTL}
In this subsection, we present the algorithm for the information-guided identification of GTL formulas.

For a pGTL formula $\varphi$, suppose that the parameter valuation $\theta=[\theta_1,\dots,\theta_z]$ belongs to a set $\Theta=[\theta^{\textrm{min}}_1,\theta^{\textrm{max}}_1]\times\dots\times[\theta^{\textrm{min}}_z,\theta^{\textrm{max}}_z]$, where $\theta^{\textrm{min}}_i\le\theta^{\textrm{max}}_i (i=1,\dots,z)$. We define the mapping $\Pi:\Theta\rightarrow[0,1]^z$ as follows:
for each $\theta\in\Theta$, $\Pi(\theta)=\omega=[\omega_1,\dots,\omega_z]$, where for each $i\in[1,z]$, we have 
\[
\begin{split}
\omega_i=                                                            
\begin{cases}
(\theta_i-\theta^{\textrm{min}}_i)/(\theta^{\textrm{max}}_i-\theta^{\textrm{min}}_i)  ~~~~~\mbox{if}~\varrho(\theta_i, \varphi) = +;\\ 
(\theta^{\textrm{max}}_i-\theta_i)/(\theta^{\textrm{max}}_i-\theta^{\textrm{min}}_i)  ~~~~~\mbox{if}~\varrho(\theta_i, \varphi) = -.
\end{cases}
\end{split}
\]
In this way, we transform the set of parameter valuations $\Theta$ to the hypercube $[0,1]^z$, and for each $i$ the pGTL formula is easier to be satisfied if $\omega_i$ is increased. Under this transformation, we use $\Omega_s$ to denote the set of parameter valuations transformed from the parameter valuations in the minimal satisfying set $\Theta_s$ (we also call $\Omega_s$ the \textit{transformed} minimal satisfying set).                                                                    

\begin{definition}
	The \textit{Hausdorff directed distance} from a set $S$ to a set $S'$ is defined as 
	\[
	\hat{d}_H(S, S')=\max_{s\in S}\min_{s'\in S'}\max_{1\le i\le z}d(s_{i}',s_{i}), 
	\]
	\label{Hausdorff}
\end{definition} 
\vspace{-10pt}
where $s=[s_1, \dots, s_z]\in S$, the directed distance $d(s_{i}',s_{i})=s_{i}-s_{i}'$, if $s_{i}>s_{i}'$; and $d(s_{i}',s_{i})=0$, if $s_{i}\le s_{i}'$ ($i=1,\dots,z$).

For a pGTL formula $\varphi$, suppose that after several queries we obtain the sets $\hat{\Omega}_{unsat}$ and $\hat{\Omega}_{sat}$ of parameter valuations, where for each $\omega\in\hat{\Omega}_{sat}$, $\varphi_{\Pi^{-1}(\omega)}$ satisfies the coverage constraint; and for each $\omega\in\hat{\Omega}_{unsat}$, $\varphi_{\Pi^{-1}(\omega)}$ violates the coverage constraint (see Fig. \ref{query_fig}). We use $\hat{\Omega}_{s}$ to denote the set of minimal parameter valuations in $\hat{\Omega}_{sat}$.                                                             

\begin{definition}($\epsilon$-approximation \cite{Julien2010})
	A set of parameter valuations $\hat{\Omega}_{s}$ is an $\epsilon$-approximation of the set $\Omega_s$ if $\hat{d}_H(\Omega_s, \hat{\Omega}_{s})\le\epsilon$, where $\hat{d}_H(\Omega_s, \hat{\Omega}_{s})$ is the Hausdorff directed distance from $\Omega_s$ to $\hat{\Omega}_{s}$.
	\label{near}
\end{definition}                                                                                                                           

\begin{proposition}\cite{Julien2010}
	A set $\hat{\Omega}_{s}$ of parameter valuations is an $\epsilon$-approximation of the transformed minimal satisfying set $\Omega_s$ if $\hat{d}_H(knee(\hat{\Omega}_{unsat}), \hat{\Omega}_{s})\le\epsilon$, where $knee(\hat{\Omega}_{unsat})$ denotes the set of \textit{knee points} of $\hat{\Omega}_{unsat}$ (a point in the boundary of $\hat{\Omega}_{unsat}$ is called a \textit{knee point} if by subtracting
	a positive number from any of its coordinates we obtain a point in the interior of $\hat{\Omega}_{unsat}$).
	\label{bd}
\end{proposition} 

The algorithm for the information-guided parametric identification of GTL formulas is shown in Algorithm \ref{inference_alg}. The identification is performed in two steps. In the first step (Line \ref{1s} to Line \ref{1e}), we approximate the minimal satisfying set $\Theta_s$ (transformed minimal satisfying set $\Omega_s$) by iteratively querying the parameter space and updating the sets of parameter valuations that satisfy and violate the coverage constraint, respectively. In the second step (Line \ref{2s} to Line \ref{2e}), we compute the parameter valuation that maximizes the avearge information gain with the parameter valuations chosen from the approximated minimal satisfying set.

We first initialize $\hat{\Omega}_{sat}$ and $\hat{\Omega}_{s}$ to be the vector of ones, and $\hat{\Omega}_{unsat}$ to be the vector of zeros. We use $r$ to denote the maximal directed distance from the set of knee points $knee(\hat{\Omega}_{unsat})$ to the set $\hat{\Omega}_{s}$ (Line \ref{max_r}), and $\omega$ to denote the knee point with the maximal directed distance to the set $\hat{\Omega}_{s}$ (Line \ref{select_alg}). After each query, we select the next query as $\omega+r/2$, which is guaranteed to lie neither in $\hat{\Omega}_{sat}$ nor in $\hat{\Omega}_{unsat}$ (we choose $\omega+r/2$ in the manner of binary search). Then we update $\hat{\Omega}_{sat}$ if $\varphi_{\Pi^{-1}(\omega+r/2)}$ satisfies the coverage constraint, and update $\hat{\Omega}_{unsat}$ if it violates the coverage constraint \cite{Julien2010} (Line \ref{update_omega}). And the same procedure repeats until an $\epsilon$-approximation of $\Omega_s$ is achieved. Finally, we identify the GTL formula with the parameter valuation in the approximated minimal satisfying set that provides the highest avearge information gain (Line \ref{2s} to Line \ref{2e}).

\begin{figure}[th]
	\centering
	\includegraphics[width=10cm]{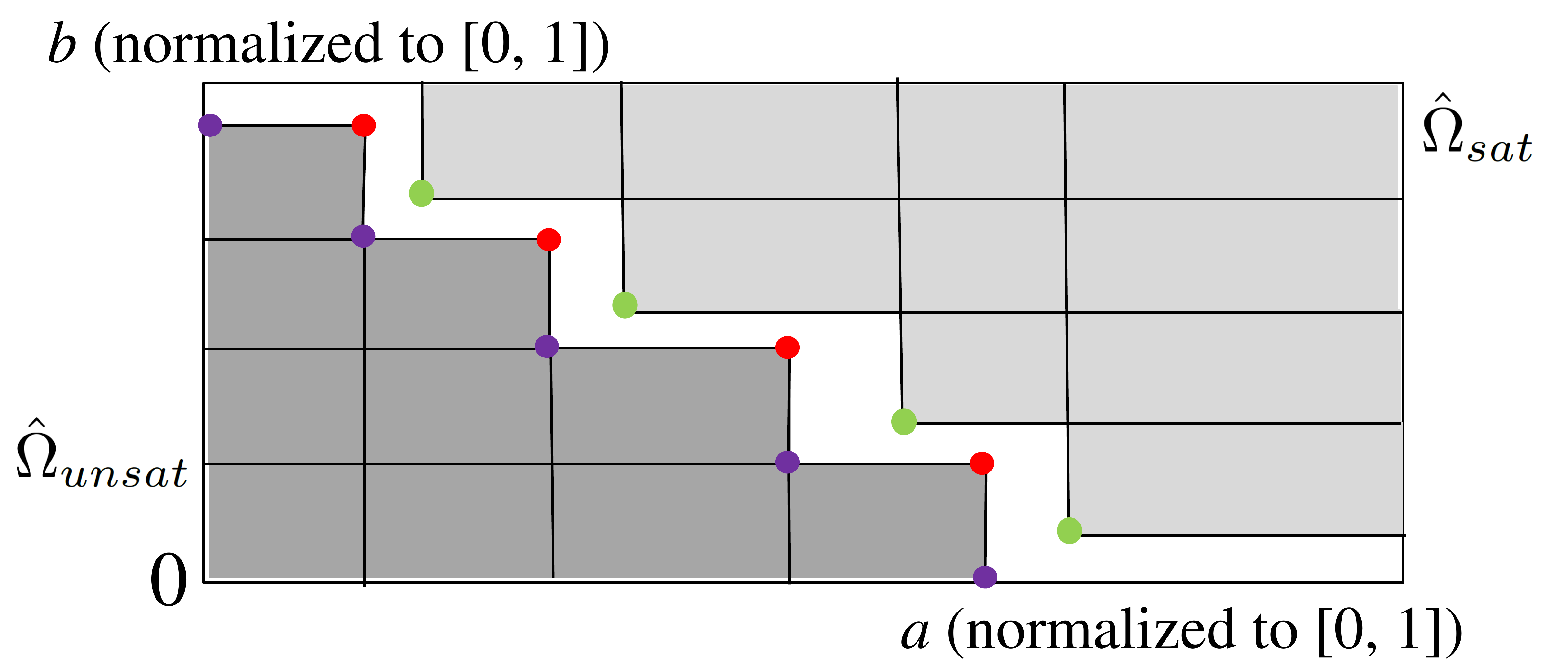}\caption{Querying the parameter space for a pGTL formula $\exists^{N}\bigcirc_{y\le a}(x\le b)$ (with $a$ and $b$ dimensions shown). Green and red points correspond to parameter valuations of GTL formulas that satisfy and violate the coverage constraint respectively, while the purple points denote the knee points.}
	\label{query_fig}
\end{figure}                                               

\begin{algorithm}
	\DontPrintSemicolon
	{
		Initialize $\hat{\Omega}_{sat}$, $\hat{\Omega}_{s}$, $\hat{\Omega}_{unsat}$, $\hat{\mathcal{I}}\gets0$\;   \label{1s}
		\While{$\hat{d}_H(knee(\hat{\Omega}_{unsat}), \hat{\Omega}_{s})>\epsilon$} 
		{$r\gets\max\limits_{\omega\in knee(\hat{\Omega}_{unsat})}d(\omega,\hat{\Omega}_{s})$\;   \label{max_r}
		$\omega\gets\argmax\limits_{\omega\in knee(\hat{\Omega}_{unsat})}d(\omega,\hat{\Omega}_{s})$, $\omega\gets\omega+r/2$\;  \label{select_alg}  \label{max_kneepoint}
		Add $\omega$ to $\hat{\Omega}_{sat}$ if $\varphi_{\Pi^{-1}(\omega)}$ satisfies the coverage constraint and add $\omega$ to $\hat{\Omega}_{unsat}$ otherwise\;   \label{update_omega}
		Update $\hat{\Omega}_{s}$ as the minimal satisfying set of $\hat{\Omega}_{sat}$  \label{1e}  }
		\For{$\omega\in\hat{\Omega}_{s}$ \label{2s}} 
	    { $\mathcal{I}\gets ComputeIG(\mathcal{S}^G_L,\varphi_{\Pi^{-1}(\omega)}, \mathcal{F}^G_L)$\;
		\If{$\mathcal{I}>\hat{\mathcal{I}}$}
		{ $\hat{\mathcal{I}}\gets\mathcal{I}$, $\hat{\omega}\gets\omega$  } \label{2e} }
		\Return{$\varphi_{\Pi^{-1}(\hat{\omega})}$}	} 
	\caption{Information-guided parameter identification of GTL formulas.}                                           
   \label{inference_alg}
\end{algorithm}	  

\section{Case Studies}
We illustrate our approaches on three studies, with Case Study 1 on the classification problem, and Case Study 2 and 3 on the identification problem. The data used in Case Study 1 and 2 are from the SLS 3D printer at UT Austin, recorded with FLIR 6701 MWIR stationary Infrared camera.

\subsection{Case Study 1}
\label{case1}

The first case study is on classifying the graph patterns of ten tensile specimens built from SLS process (see Fig. \ref{classify_fig}). The tensile specimens have varying strengths, where the five stronger specimens labeled 1 have tensile strength above 46 MPa and the other five labeled -1 have tensile strength below 34 MPa. 
\begin{figure}[th]
	\centering
	\includegraphics[width=6cm]{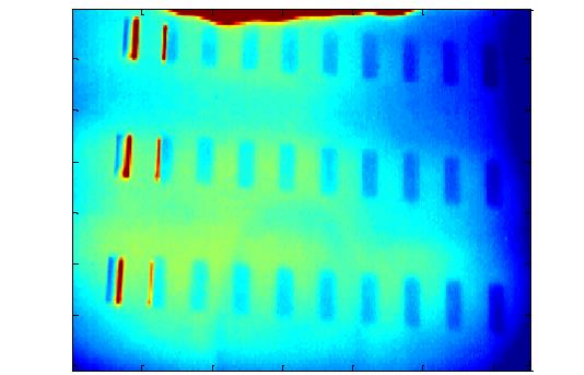}\caption{Infrared image of the rectangular cross-sections of the tensile specimens on the surface of the SLS powder bed \cite{sam_thesis}.}
	\label{classify_fig}
\end{figure}         

We partition the fill region of each tensile specimen into 20 subregions (1200$\times$1200 \SI{}{\micro\metre}$^2$ for each subregion, with 210 layers), where each subregion is considered a node of a fully connected graph. The edge label $y$ represents the Euclidean distance between the nodes (1 unit represents 1200 \SI{}{\micro\metre}). As different layers are sintered at evenly spaced time instants, we use the layer indices to represent the time indices.
                              
We use the following templates for type-I pGTL formulas (while other valid pGTL formulas can be also added to the set of templates, we choose the following ones as they are simple and sufficiently expressive for our applications):  
	\vspace{-5pt}
\begin{align}\nonumber
\begin{split}
\mathcal{P}^{\textrm{I}}=\{&\Box_{\ge i_1, \le i_2}\exists^{N}\bigcirc_{\rho}\pi,~~~~~~~~~~
\Diamond_{\ge i_1, \le i_2}\exists^{N}\bigcirc_{\rho}\pi,\\
&\Box_{\ge i_1, \le i_2}\Diamond_{\le i_3}\exists^{N}\bigcirc_{\rho}\pi,~~~~
\Diamond_{\ge i_1, \le i_2}\Box_{\le i_3}\exists^{N}\bigcirc_{\rho}\pi,\\
&\Box(\pi_1\Rightarrow\Box_{\le i}\exists^{N}\bigcirc_{\rho}\pi_2),~~~\Box(\pi_1\Rightarrow\Diamond_{\le i}\exists^{N}\bigcirc_{\rho}\pi_2)\},
\end{split}
\end{align}
where $\pi$, $\pi_1$ and $\pi_2$ are atomic node propositions in the form of $x\ge c_1$ or $x\le c_1$ ($c_1\in\mathbb{R}$), $\rho$ is an edge proposition in the form of $y\le c_2$ ($c_2$ is a positive integer), $N$ is a positive integer, $i_1, i_2,i_3, i\in\mathbb{T}$ and $i_1<i_2$.                    

We use the following templates for type-II pGTL formulas:
\begin{align}\nonumber
\begin{split}
\mathcal{P}^{\textrm{II}}=\{& \exists^{N}\bigcirc_{\rho}\Box_{\ge i_1, \le i_2}\pi,~~~~~~~~
\exists^{N}\bigcirc_{\rho}\Diamond_{\ge i_1, \le i_2}\pi,\\
&\exists^{N}\bigcirc_{\rho}\Box_{\ge i_1, \le i_2}\Diamond_{\le i_3}\pi,~~
\exists^{N}\bigcirc_{\rho}\Diamond_{\ge i_1, \le i_2}\Box_{\le i_3}\pi
\},
\end{split}
\end{align}
where $\pi$, $\rho$, $N$, $i_1$, $i_2$ and $i_3$ are as described in $\mathcal{P}^{\textrm{I}}$.                        

We set $m_{\rm{th}}=0.02$, $\eta_{\rm{th}}=3$ and $\hat{m}_{\rm{th}}=0.1$. Using the approach illustrated in Sec. \ref{classify_sec}, we obtain the following GTL formula from $\mathcal{P}^{\textrm{I}}$ and $\mathcal{P}^{\textrm{II}}$ with zero nodal misclassification rate:
\begin{align}\nonumber
\begin{split}
\varphi_{\theta^{\ast},1}^{\ast}=&\Box\big(T\ge181.1\Rightarrow\Box_{\le6}\exists^{8}\bigcirc_{y\le2}T\le198.0\big)\wedge\\&\Box\big(T\le204.0\Rightarrow\Box_{\le2}\exists^{3}\bigcirc_{y\le1}T\ge181.8\big),
\end{split}
\end{align}
which means ``(for any node) whenever the temperature is at least 181.1 $\textrm{C}^{\circ}$, then for the next 6 layers there are always at least 8 nodes within distance of 2 where the temperature is at most 198.0 $\textrm{C}^{\circ}$; whenever it is at most 204.0 $\textrm{C}^{\circ}$, then for the next 2 layers there are always at least 3 nodes within distance of 1 where the temperature is at least 181.8 $\textrm{C}^{\circ}$''.

$\varphi_{\theta^{\ast},1}^{\ast}$ is validated with another set of ten tensile specimens (five labeled 1 with tensile strength above 43 MPa and five labeled -1 with tensile strength below 37 MPa), with nodal misclassification rate of 8.33\%.

\subsection{Case Study 2}   
\label{case2}
The second case study is on identifying informative patterns from data of SLS cooldown process (see Fig. \ref{color_intro} (a) in Sec. I). We record 16 graph temporal trajectories from $7\times7$ grids of the powder bed, where each of the 49 cells (400$\times$400 \SI{}{\micro\metre}$^2$ for each cell) is considered a node of a fully connected graph. The edge label $y$ represents the Euclidean distance between the nodes (1 unit represents 400 \SI{}{\micro\metre}). 

We set $p_{\rm{th}}=0.98$  and $\epsilon=0.05$. Through Algorithm \ref{inference_alg}, we obtain the best type-I and  type-II GTL formulas from $\mathcal{P}^{\textrm{I}}$ and $\mathcal{P}^{\textrm{II}}$ as (with the average information gain of 0.1563 and 0.0013, respectively, both with coverage rate of 100\%):
\[
\begin{split}
\varphi_{\theta^{\ast},2}^{\textrm{I}\ast}&=\Box\big(T\ge183.4\Rightarrow\Box_{\le3}\exists^{2}\bigcirc_{y\le1}T\le182.8\big),\\
\varphi_{\theta^{\ast},2}^{\textrm{II}\ast}&=\exists^{4}\bigcirc_{y\le2}\Box_{\ge3,\le8}(T\ge178.7),
\end{split}  
\]
where $\varphi_{\theta^{\ast},2}^{\textrm{I}\ast}$ means ``(for any node) whenever the temperature is at least 183.4 $\textrm{C}^{\circ}$, then for the next 3 time steps there are always at least 2 nodes within distance of 1 where the temperature is at most 182.8 $\textrm{C}^{\circ}$'', and $\varphi_{\theta^{\ast},2}^{\textrm{II}\ast}$ reads as ``(for any node) there exists at least 4 nodes within distance of 2 where the temperature is always at least 178.7 $\textrm{C}^{\circ}$ from time step 3 to 8'' (each time step lasts for 33 milliseconds).

$\varphi_{\theta^{\ast},2}^{\textrm{I}\ast}$ and $\varphi_{\theta^{\ast},2}^{\textrm{II}\ast}$ are validated with another set of 16 graph temporal trajectories recorded from another layer of the powder bed, both with coverage rate of 100\%.

\subsection{Case Study 3}   
\label{case3}
The third case study is on identifying informative patterns from simulation data of a swarm of robots. We partition the workspace into 9 subregions (as shown in Fig. \ref{swarm_fig}), where each subregion is considered a node of a fully connected graph. The edge label $y$ represents the Euclidean distance between the centroids of the subregions. The probabilistic densities of the robots in the subregions are governed by a time-varying Markov chain \cite{Demirer2018}. 

\begin{figure}[th]
	\centering
	 \includegraphics[width=8cm]{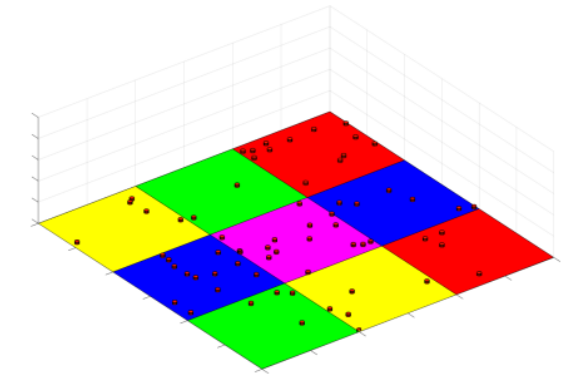}\caption{The swarm of 72 robots in the 9 sub-regions.}
	\label{swarm_fig}
\end{figure}      

We randomly generate graph-temporal trajectories and randomly choose 10 from them that satisfy the following constraint: whenever the probabilistic density of a subregion reaches above 1/8, then for the next 2 time units there always exists at least one neighbor subregion within distance of 1 with probabilistic density below 1/9. Then we infer a GTL formula from the 10 graph-temporal trajectories using Algorithm \ref{inference_alg}. 

We set $p_{\rm{th}}=0.98$  and $\epsilon=0.05$. Through Algorithm \ref{inference_alg}, we obtain the best type-I and type-II GTL formulas from $\mathcal{P}^{\textrm{I}}$ and $\mathcal{P}^{\textrm{II}}$ as (with the average information gain of 0.1 and 0.0043, respectively, both with coverage rate of 100\%): 
\[
\begin{split}
\varphi_{\theta^{\ast},3}^{\textrm{I}\ast}&=\Box\big(x\ge0.1894\Rightarrow\Box_{\le2}\exists^{1}\bigcirc_{y\le1}x\le0.1137\big),\\
\varphi_{\theta^{\ast},3}^{\textrm{II}\ast}&=\exists^{2}\bigcirc_{y\le1}\Box_{\ge4,\le6}(x\ge0.0379), 
\end{split}  
\]                 
where $x$ is the probabilistic density in a subregion. It can be seen that $\varphi_{\theta^{\ast},3}^{\textrm{I}\ast}$ is different but similar with the set constraint.

$\varphi_{\theta^{\ast},3}^{\textrm{I}\ast}$ and $\varphi_{\theta^{\ast},3}^{\textrm{II}\ast}$ are validated with another set of 10 randomly generated graph temporal trajectories that satisfy the set constraint, both with coverage rate of 100\%.

\section{Conclusion}
We have introduced GTL and proposed the framework and algorithms to infer GTL formulas from data for classification and identification. 
For future work, we will consider more efficient methods for inferring more general forms of GTL formulas. In various network systems, specifications can be expressed in GTL, hence verification and controller synthesis can be also conducted with GTL specifications. 

\bibliographystyle{IEEEtran}
\bibliography{zherefclean_submit}

\end{document}